\documentclass[12pt,reqno]{article}

\usepackage[affil-it]{authblk}

\usepackage[usenames]{color}
\usepackage{graphicx}
\usepackage{color}
\usepackage{amscd,amsmath,amssymb,amsthm}

\usepackage{fullpage}
\usepackage{float}

\newtheorem{theorem}{Theorem}[section]

\newtheorem{proposition}[theorem]{Proposition}

\theoremstyle{definition}

\newtheorem{example}[theorem]{Example}

\theoremstyle{remark}
\newtheorem{remark}[theorem]{Remark}

\numberwithin{equation}{section}

\begin{document}

\title{\LARGE\textsc{A General Formula for the\\ Generation Time}}

\vskip 1cm

\author[1]{Fran\c{c}ois Bienvenu}
\author[2,3]{Lloyd Demetrius}
\author[1]{ St\'{e}phane Legendre \thanks{Corresponding author: legendre@ens.fr}}
\affil[1]{Team of Mathematical Eco-Evolution, Ecole Normale Sup\'{e}rieure, Paris, France}
\affil[2]{Department of Organismic and Evolutionary Biology, Harvard University,  USA}
\affil[3]{Max Planck Institute for Molecular Genetics, Berlin, Germany}


\date{10/06/2013}

\maketitle

\begin{abstract}
We show that the generation time -- a notion usually described in a biological context -- can be defined in a general way as a return time in a conveniently constructed finite Markov chain. The simple formula we obtain agrees with previous results derived for structured populations projected in discrete time, and allows to define the generation time of any process described by a weighted directed graph whose matrix is primitive.
\end{abstract}

\section{Introduction}\label{introduction}

The generation time $T$ is a biological notion, intuitively thought of as the time between two generations. In age-classified population models, $T$ is the mean age of the mothers at the birth of their offspring. 

However, for many organisms, the pertinent stages of life history are not classified by age but by other biological parameters, like size. This is the case in many animal species (e.g., in reptiles, fishes, and arthropods) and most often the case in plants. For example, plant phenotypes may stay in a size class for an indefinite length of time until favorable light conditions allow them to grow and get to another size class. Moreover, on top of the main stage-descriptive parameter (age, size), population dynamics models often consider stages indexed by other descriptive variables (morphotype, behavioral category). For example, in metapopulation models, the sites of the local populations connected by migrations are indexing the age- or size- classes.

As the generation time is a fundamental biological descriptor, linked with the cycle time of biochemical reactions within the organism \cite{Demetrius_2006,DemetriusLegendre_2009,Demetrius_2013}, it is desirable to compute it for a large class of models, i.e., for life cyles more complex than age-classified ones. This has been done notably by Lebreton \cite{Lebreton_1996} for age-classified metapopulation models, and by Cochran and Ellner \cite{CochranEllner_1992} for complex life cycles. These articles describe the computation of many useful demographic descriptors, but the formulas for $T$ are rather complicated.

A multicellular organism can be envisioned as a set of germ cells (producing the gametes), that are carried by a body made of somatic cells -- the vehicle of the germ cells. The generation time can then be seen as the mean time by which the germ cells shift vehicle. $T$ can therefore be defined as the mean time of first return to the transition leading to the creation of a novel organism (by fusion of the gametes in sexual organisms). We call \textit{reproductive} such a transition.

In this study, we develop a general setting for computing the generation time $T$ as the mean time of first return in a conveniently constructed finite Markov chain. The formula is surprisingly simple:
\begin{displaymath}
T = \frac{1}{\displaystyle\sum_{[j \to i] \in R} e_{ij}}.
\end{displaymath} 
In this expression, $R$ is the set of reproductive arcs $j \to i$ in the weighted directed graph representing the life cycle of the population, and the $e_{ij}$'s are the elasticities of the matrix associated with the digraph.

\section{Matrix population models}\label{matrixpop} 
In its life history, an organism goes through different stages, typically developmental stages toward a mature form, followed by reproductive stages where offspring is produced. A population is considered as a set of identical individuals sharing the same life cycle, parameterized by average demographic parameters (survival, fecundity).

Matrix population models \cite{Caswell_2000} allow to project in discrete time $t = 0, 1,2, \dots$ populations structured by age, size or other classifying parameters. In this framework, an individual within the population is represented by the life cycle graph, a weigthed directed graph $\mathcal{A}$. The nodes of $\mathcal{A}$ represent the stages traversed by individuals during their lives and the arcs are weigthed by the demographic parameters (Fig. \ref{fig_leslie}). The weight $a_{ij}$ associated with the arc $j \to i$ represents the contribution of stage $j$ at time $t$ to stage $i$ at time $t+1$. We note that the adjacency matrix of the weighted digraph  $\mathcal{A}$ is the transpose $^t\mathbf{a}$. The population matrix $\mathbf{a} = (a_{ij})$ allows to compute the number of individuals in the stages from one time step to the next:
\begin{equation}\label{popi}
n_i(t+1) = \sum_{j} a_{ij} n_j(t). 
\end{equation} 
Here $n_i$ is the number of individuals in stage $i$. In matrix form,
\begin{displaymath}
\mathbf{n}(t+1) = \mathbf{a}\mathbf{n}(t),
\end{displaymath} 
with $\mathbf{n}(t) = (n_i(t))$ the population vector at time $t$.

We assume that the non negative population matrix $\mathbf{a}$ is primitive (irreducible and aperiodic). Under this assumption, the Perron-Frobenius theorem garantees the existence of a dominant eigenvalue $\lambda > 0$, a real eigenvalue that is largest in modulus than all other eigenvalues  \cite{Seneta_2006}. The dominant eigenvalue $\lambda$ is the asymptotic growth rate of the population. Indeed, for large $t$, the total population size at time $t$, $n(t) = \sum_i n_i(t)$, verifies
\begin{displaymath} 
n(t+1) \sim \lambda n(t). 
\end{displaymath} 
The left and right eigenvectors of $\mathbf{a}$ with respect to $\lambda$, $\mathbf{v}$ and $\mathbf{w}$, have positive entries. The right eigenvector $\mathbf{w} = (w_i)$, when normalized $\sum_i w_i=1$, is called the \textit{stable stage distribution}. Indeed, for all $i$, the proportion of individuals in stage $i$ at time $t$ converges toward $w_i$:
\begin{equation}\label{wi}
\frac{n_i(t)}{n(t)} \to w_i.
\end{equation}  
In the Leslie model \cite{Leslie_1945}, the stages correspond to age classes. The first row of the population matrix is constituted of fertility transitions ($f_i$) leading to the newborn stage 1, and the subdiagonal contains survival transitions ($s_i$) (Fig. \ref{fig_leslie}). The growth rate $\lambda$ is the largest root of the characteristic equation 
\begin{displaymath} 
f(X) = \sum_{i} l_i f_i X^{-i} = 1, 
\end{displaymath} 
where $l_1 = 1$ and $l_i = s_1 \ldots s_{i-1}$ for $i>1$ \cite{Euler_1760,Caswell_2000}. Therefore, the quantities $q_i = l_i f_i \lambda^{-i}$ verify $\sum_i q_i = 1$ and form a distribution. The mean of this distribution,
\begin{equation}\label{gentime1}
T = \sum_{i} i l_i f_i \lambda^{-i}, 
\end{equation} 
is the generation time, interpreted as the mean age of the mothers at the birth of their daughters when the population is at the stable age distribution.

\begin{figure}[ht]
\setlength{\unitlength}{0.8cm}
\begin{picture}(14,6)(0,0)
\put(0,0){\includegraphics[scale=0.34]{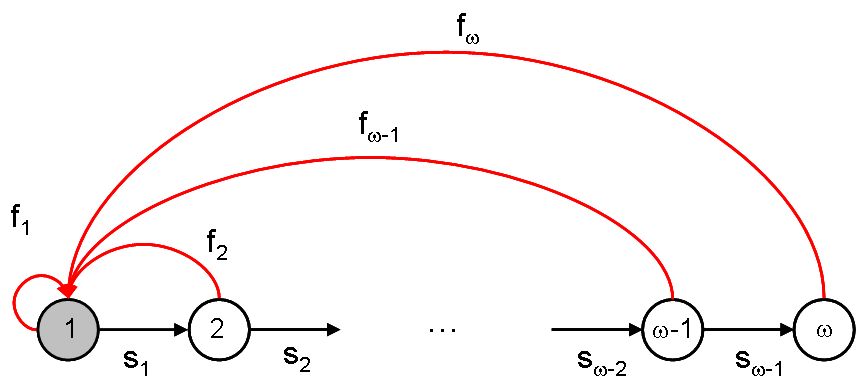}}
\put(14.5,2)
{$\begin{bmatrix}
f_1    & f_2    & \cdots & f_{\omega-1} & f_\omega\\
s_1    & 0      & \cdots & 0            & 0\\
0      & s_2    & \cdots & 0            & 0\\
\vdots &        &        &              &\vdots\\
0      &        & \cdots & s_{\omega-1} & 0
\end{bmatrix}$}
\end{picture}
\caption{The life cycle graph of the age-classified Leslie model and the corresponding matrix. Stages $1,2,\ldots ,\omega$ correspond to individuals aged $1,2,\ldots ,\omega$. Reproductive transitions are in red and point to the newborn stage 1. $s_i$ is the survival rate from age $i$ to age $i+1$ and $f_i$ is the fertility rate of individuals aged $i$. When age at first reproduction $\alpha > 1$, the fertilities $f_1, \ldots, f_{\alpha-1}$ are set to 0.}
\label{fig_leslie}
\end{figure}

\section{The Markov matrix}\label{markovmatrix} 
To perform random walks in the digraph  $\mathcal{A}$, and compute the generation time as a return time, we first normalize the matrix $\mathbf{a}$ into a matrix $\mathbf{p}$ so that the weights $p_{ij}$ associated with the out-arcs $i \to j$ from any stage $i$ sum to 1: $\sum_j p_{ij} = 1$. Thus $\mathbf{p}$ will be a Markov matrix.

The convenient Markov matrix $\mathbf{p}$, see \cite{Demetrius_1974,Demetrius_1975,Tuljapurkar_1982}, is given by
\begin{equation}\label{pij}
p_{ij} = \frac{a_{ij}w_j}{\lambda w_i}.
\end{equation}
We check that the rows of $\mathbf{p}$ sum to 1: as $\mathbf{w}$ is a right eigenvector of $\mathbf{a}$, 
\begin{displaymath}
\sum_j a_{ij}w_j = \lambda w_i,
\end{displaymath}
we have
\begin{displaymath}
\sum_j p_{ij} = \frac{1}{\lambda w_i} \sum_j a_{ij}w_j = 1.
\end{displaymath}
An informal argument for formula (\ref{pij}) is the following. There are $n_i(t)$ individuals in stage $i$ at time $t$, $a_{ij} n_j(t-1)$ of which come from stage $j$. Hence the probability of coming from stage $j$ for an individual in stage $i$ is 
\begin{displaymath}
\frac{a_{ij} n_j(t-1)}{n_i(t)}. 
\end{displaymath}
Assuming the population at the stable stage distribution, we have from (\ref{wi}), $n_j(t-1) \sim w_j n(t-1)$ and $n_i(t) \sim \lambda w_i n(t-1)$, giving (\ref{pij}). Moreover, the apparent absence of transposition in (\ref{pij}) comes from the fact that two transpositions have actually been performed: one to reverse the time as we explore the genealogy, the second to switch from the formalism of matrix population models to that of Markov chains.

To the digraph $\mathcal{A}$ is associated the digraph $\mathcal{P}$ that has the same nodes and arcs as $\mathcal{A}$, but where the weight associated with the arc $i \to j$ is $p_{ij}$, i.e., $\mathbf{p}$ is the adjacency matrix of $\mathcal{P}$.

\begin{proposition}\label{stationary_p}
The stationary distribution $\boldsymbol{\pi}$ of the Markov chain $\mathbf{p}$ associated with the population matrix $\mathbf{a}$ is given by 
\begin{equation}\label{pi_j}
\pi_j = \frac{v_j w_j}{\mathbf{v} \mathbf{w}},
\end{equation}
where $\mathbf{v} = (v_j)$ and $\mathbf{w} = (w_j)$ are left and right eigenvectors of $\mathbf{a}$ with respect to the dominant eigenvalue.
\end{proposition}
\begin{proof}
The entries $\pi_j$ of the row vector $\boldsymbol{\pi}$ sum to 1 because of the normalization by the dot product
\begin{displaymath}
\mathbf{v}\mathbf{w} = \sum_j v_jw_j.
\end{displaymath}
We show that $\boldsymbol{\pi}$ is a left eigenvector of $\mathbf{p}$ with respect to the eigenvalue 1. Using the relation $\mathbf{v}\mathbf{a} = \lambda \mathbf{v}$, we have
\begin{displaymath}
(\mathbf{v}\mathbf{w}) \sum_i \pi_i p_{ij} = \sum_i v_i w_i p_{ij} = \sum_i v_i w_i \frac{a_{ij}w_j}{\lambda w_i} = \frac{w_j}{\lambda} \sum_i v_i a_{ij} = v_j w_j = (\mathbf{v} \mathbf{w}) \pi_j.
\end{displaymath}
Hence $\boldsymbol{\pi}\mathbf{p} = \boldsymbol{\pi}$, and $\boldsymbol{\pi}$ is the stationary distribution of the Markov chain \cite{Seneta_2006}.
\end{proof}

As the digraph $\mathcal{P}$ is irreducible, there is a directed path from any stage to any other. In particular, any stage is contained in a directed cycle. The mean time of first return to stage $j$ is given by
\begin{equation}\label{1surpi1}
\tau_j = \frac{1}{\pi_j}.
\end{equation}
More generally, if $B$ is a subset of the set of stages, the mean time of first return to $B$ is
\begin{equation}\label{1sursompi}
\tau_B = \frac{1}{\displaystyle\sum_{j \in B}^{} \pi_j}.
\end{equation}
Indeed, the quantity $\sum_{j \in B} \pi_j$ is the asymptotic proportion of particles in the stages $j \in B$ after they have traversed the digraph according to the probabilities of $\mathbf{p}$. The inverse of this quantity is the mean time of first return of the particles to $B$. 

\begin{example}\label{example_leslie}
For the Leslie matrix, left and right eigenvectors $\mathbf{v}$ and $\mathbf{w}$ are given by
\begin{displaymath}
v_i = \frac{\lambda^{i-1}}{l_i}\sum_{j=1}^{\omega} l_j f_j \lambda^{-j},  \quad w_i = l_i \lambda^{-(i-1)}, \quad i = 1, \ldots, \omega.
\end{displaymath}
Here $\omega$ is the total number of age classes. After some algebra, 
\begin{displaymath}
\mathbf{v} \mathbf{w} = \sum_i i l_i f_i \lambda^{-i} = T.
\end{displaymath}
As (\ref{pi_j}) is independent of the scaling of the eigenvectors, we may rescale $\mathbf{v}$ and $\mathbf{w}$ so that $v_1 = 1$ and $w_1 =1$, and we recover (\ref{1surpi1}):
\begin{displaymath}
T = \frac{1}{\pi_1}.
\end{displaymath}
\end{example}

\section{Reproductive arcs}\label{reproarcs} 
Newborn stages are characterized by the fact that they are entered by reproductive arcs. By \textit{reproductive arc}, we mean any transition $j \to i$ in the life cycle graph such that individuals in stage $j$ contribute to \textit{novel} individuals in stage $i$. These arcs are determined by the biology of the species. In animals, reproductive arcs are usually identified as transitions leading to the production of offspring, but organisms may also reproduce by fission, as is often the case in plants. In this case, the corresponding arcs may be considered as reproductive or not depending on the biological question.

Newborn stages are ambiguously defined because non reproductive arcs (survival, migration) can also lead to such stages. This is for example the case in the life cycle graph of the teasel \textit{Dipsacus sylvestris} \cite{Caswell_1978b,CochranEllner_1992,Caswell_2000} (Fig. \ref{fig_teasel}). In this annual plant, the arc \textit{Small rosette} $\to$ \textit{Medium rosette} together with a self-loop enter the \textit{Medium rosette} stage. These transitions  correspond to a probability of changing size class and of staying in the size class, respectively, and are not reproductive. Nevertheless, the stage \textit{Medium rosette} is a newborn stage because of the reproductive transition \textit{Flowering plant} $\to$ \textit{Medium rosette}.

\begin{figure}
\begin{center}
\includegraphics[scale=0.4]{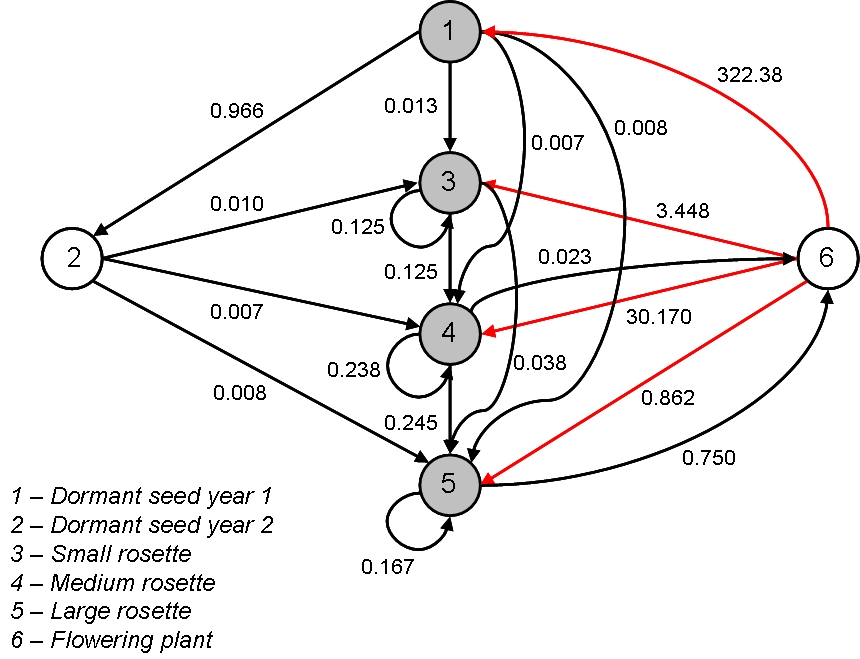}
\caption{The life cycle graph of the teasel \textit{Dipsacus sylvestris}. Newborn stages are in grey and reproductive arcs are in red.}
\label{fig_teasel}
\end{center}
\end{figure} 

Thus, we wish to compute the generation time not as the mean time of first return to a newborn stage, but as the mean time of first return to a reproductive arc. To this end, we construct a convenient weighted digraph $\tilde{\mathcal{P}}$ from the digraph $\mathcal{P}$. The nodes of $\tilde{\mathcal{P}}$ are the arcs $i \to j$ of $\mathcal{P}$, and we create an arc between 2 nodes $a = [i \to j]$, $b = [k \to l]$ of $\tilde{\mathcal{P}}$ if and only if $k = j$, in which case the weight
\begin{displaymath}
\tilde{p}_{ab} = \tilde{p}_{i \to j, j \to l} = p_{jl}
\end{displaymath} 
is associated with the arc $a \to b$ (Fig. \ref{fig_ptilde}). In other words, there is an arc joining $a = [i \to j]$ to $b = [k \to l]$ in $\tilde{\mathcal{P}}$ if and only if $a$ is an in-arc of node $j$ and $b$ an out-arc of this node in $\mathcal{P}$ (thus $k = j$). The weight associated with $a \to b$ in $\tilde{\mathcal{P}}$ is the weight $p_{jl}$ associated with the out-arc $j \to l$ in $\mathcal{P}$. By construction, in $\tilde{\mathcal{P}}$, all arcs entering the node $b = [j \to l]$ bear the same weight $p_{jl}$.

The adjacency matrix of the digraph $\tilde{\mathcal{P}}$ is denoted $\tilde{\mathbf{p}} = (\tilde{p}_{ab})$.

\begin{figure}
\begin{center}
\includegraphics[scale=0.4]{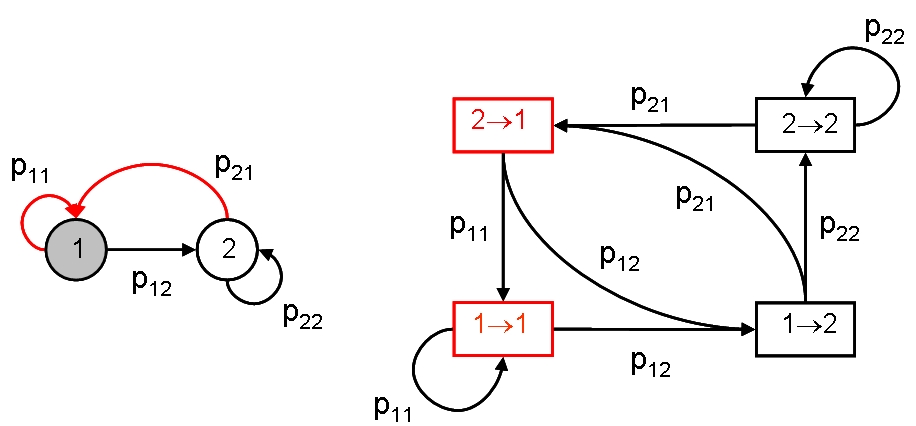}
\caption{The digraph $\mathcal{P}$ associated with a 2 stages life cycle, and the corresponding digraph $\tilde{\mathcal{P}}$ whose nodes are the arcs of  $\mathcal{P}$. Reproductive arcs are in red.}
\label{fig_ptilde}
\end{center}
\end{figure} 
 
\begin{proposition}\label{ptilde}
\noindent
\begin{itemize}
\item The matrix $\tilde{\mathbf{p}}$ is a Markov matrix that is irreducible (primitive) if and only if the Markov matrix $\mathbf{p}$ is irreducible (primitive).
\item The return time to the transition $i \to j$ in the digraph $\mathcal{P}$ is the same as the return time to the node $a = [i \to j]$ in the digraph $\tilde{\mathcal{P}}$.
\end{itemize}
\end{proposition}
\begin{proof}
Let the index $a = [i \to j]$ be fixed. As the entries of the row of index $j$ of  $\mathbf{p}$ sum to 1, we have 
\begin{displaymath}
\sum_b \tilde{p}_{ab} = \sum_l \tilde{p}_{i \to j, j \to l} = \sum_l p_{jl} = 1.
\end{displaymath}
Hence, the entries of the row of index $a$ of $\tilde{\mathbf{p}}$ sum to 1, and $\tilde{\mathbf{p}}$ is a Markov matrix.

Let $a = [i \to j]$ and $b = [k \to l]$ be any 2 nodes in $\tilde{\mathcal{P}}$. Since $\mathbf{p}$ is irreducible, there is a directed path from $j$ to $k$ in $\mathcal{P}$, say
\begin{equation} \label{path_p}
j \to j_1 \to j_2 \to \cdots \to j_{m-1} \to k.
\end{equation}
By irreducibility of $\mathcal{P}$, $j$ belongs to a cycle, so that there exists an arc $i \to j$ for some node $i \neq j$. In $\tilde{\mathcal{P}}$, we now have the path
\begin{equation} \label{path_ptilde}
a = [i \to j] \to [j \to j_1] \to [j_1 \to j_2] \to \cdots \to [j_{m-1} \to k] \to [k \to l] = b.
\end{equation}
This shows that $\tilde{\mathbf{p}}$ is irreducible. Moreover, the path (\ref{path_ptilde}) as the same length $m$ and the same weights as the original path (\ref{path_p}). 
Conversely, from a path (\ref{path_ptilde}) in $\tilde{\mathcal{P}}$ we can construct a path (\ref{path_p}) in $\mathcal{P}$ that has the same length and weights as the path (\ref{path_ptilde}). 
The second assertion of Proposition \ref{ptilde} is a direct consequence of the fact that paths in $\mathcal{P}$ and $\tilde{\mathcal{P}}$ have the same lengths and weights (though they are not in the same number). As a result, the greatest common divisor of cycle lengths in $\mathcal{P}$ and $\tilde{\mathcal{P}}$ are equal, implying the equivalence of $\mathbf{p}$ primitive and $\tilde{\mathbf{p}}$ primitive.  
\end{proof}

\begin{proposition}\label{stationary_ptilde}
The stationary distribution of the Markov chain $\tilde{\mathbf{p}}$ associated with the Markov chain $\mathbf{p}$ is given by
\begin{equation}\label{wa}
\tilde{\boldsymbol{\pi}} = (\tilde{\pi}_b), \quad \tilde{\pi}_b = \tilde{\pi}_{j \to l} = \pi_j p_{jl},
\end{equation}
where $\boldsymbol{\pi} = (\pi_j)$ is the stationary distribution of $\mathbf{p}$. 
\end{proposition}
\begin{proof}
The entries $\tilde{\pi}_b$ of the row vector $\tilde{\boldsymbol{\pi}}$ sum to 1 because $\boldsymbol{\pi}$ is a left eigenvector of $\mathbf{p}$: 
\begin{displaymath}
\sum_b \tilde{\pi}_{b} = \sum_{j,l} \tilde{\pi}_{j \to l} = \sum_{j,l} \pi_j p_{jl} = \sum_l \left( \sum_j \pi_j p_{jl} \right) = \sum_l \pi_l = 1.
\end{displaymath}
Let the index $b = [j \to l]$ be fixed. Then,
\begin{displaymath}
\sum_a \tilde{\pi}_a \tilde{p}_{ab} = \sum_{i} \tilde{\pi}_{i \to j} \tilde{p}_{i \to j,j \to l} = \sum_i \pi_i p_{ij}p_{jl} = \left( \sum_i \pi_i p_{ij} \right) p_{jl} = \pi_j p_{jl} = \tilde{\pi}_b,
\end{displaymath}
so that $\tilde{\boldsymbol{\pi}}$ is a left eigenvector of $\tilde{\mathbf{p}}$.
\end{proof}

\section{Generation time}\label{gentime} 
Before giving the main result of this study, we recall tools of perturbation analysis. For a primitive matrix $\mathbf{a}$ with dominant eigenvalue $\lambda$, the \textit{sensitivity} of $\lambda$ to changes in the parameter $x$ is
\begin{displaymath}
s_\lambda(x) = \frac{\partial \lambda}{\partial x}.
\end{displaymath}
The sensitivity of $\lambda$ to changes in the matrix entry $a_{ij}$ is given by \cite{Demetrius_1969,Caswell_1978a}
\begin{displaymath}
s_{ij} = \frac{\partial \lambda}{\partial a_{ij}} = \frac{v_i w_j}{\mathbf{v} \mathbf{w}}.
\end{displaymath}
The \textit{elasticity} of $\lambda$ to changes in $x$ quantifies proportional changes:
\begin{displaymath}
e_\lambda (x) = \frac{x}{\lambda} \frac{\partial \lambda}{\partial x},
\end{displaymath}
i.e., if $x$ changes by a proportion $\alpha$ then $\lambda$ changes by the proportion $\alpha e_\lambda(x)$. The elasticity of $\lambda$ to changes in the matrix entry $a_{ij}$ is then given by
\begin{equation}\label{elas_ij}
e_{ij} = \frac{a_{ij}}{\lambda} \frac{\partial \lambda}{\partial a_{ij}} = \frac{a_{ij}}{\lambda} \frac{v_i w_j}{\mathbf{v} \mathbf{w}}.
\end{equation}

\begin{theorem}\label{main}
Let $\mathbf{a}$ be a non negative primitive matrix associated with a weighted directed graph $\mathcal{A}$ so that $^t\mathbf{a}$ is the adjacency matrix of $\mathcal{A}$. Let R be a subset of the set of arcs of $\mathcal{A}$. Then the mean time of first return to R is
\begin{equation}\label{T_R}
T_R = \frac{1}{\displaystyle\sum_{[j \to i] \in R} e_{ij}},
\end{equation} 
where $e_{ij}$ is the elasticity of the dominant eigenvalue of $\mathbf{a}$ to changes in the entry $a_{ij}$.

In particular, if $R$ is the set of reproductive arcs of $\mathcal{A}$, then the generation time $T$ is given by
\begin{equation}\label{T}
T = \frac{1}{\displaystyle\sum_{[j \to i] \in R} e_{ij}}.
\end{equation} 
\end{theorem}
\begin{proof}
Arc $[j \to i]$ in $\mathcal{A}$ corresponds to arc $[i \to j]$ in $\mathcal{P}$. By (\ref{1sursompi}) and Proposition \ref{ptilde}, the mean time of first return to the arcs of $R$ is 
\begin{equation}\label{1surpi}
\frac{1}{\displaystyle\sum_{[j \to i] \in R} \tilde{\pi}_{i \to j}}.
\end{equation}
Using (\ref{pij}), (\ref{wa}), (\ref{elas_ij}), we have 
\begin{equation}\label{eij}
\tilde{\pi}_{i \to j} = \pi_i p_{ij} = \frac{a_{ij}w_j}{\lambda w_i} \frac{v_i w_i}{\mathbf{v} \mathbf{w}} = \frac{a_{ij}}{\lambda} \frac{v_i w_j}{\mathbf{v} \mathbf{w}} = e_{ij}.
\end{equation}
\end{proof}

\begin{example}\label{example_teasel}
For \textit{Dipsacus sylvestris} (Fig. \ref{fig_teasel}), the matrix of elasticities is
\begin{displaymath}
\begin{bmatrix}
0      & 0      & 0      & 0      & 0      & \underline{0.0667}\\
0.0007 & 0      & 0      & 0      & 0      & 0\\
0.0238 & 0.0007 & 0.0004 & 0      & 0      & \underline{0.0045}\\
0.0073 & 0      & 0.0025 & 0.0271 & 0      & \underline{0.2285}\\
0.0563 & 0      & 0.0051 & 0.1875 & 0.0226 & \underline{0.0439}\\
0      & 0      & 0      & 0.0509 & 0.2928 & 0
\end{bmatrix} 
\end{displaymath}
where the underlined entries correspond to the 4 reproductive transitions. The generation time is
\begin{displaymath}
T = \frac{1}{0.0667 + 0.0045 + 0.2285 + 0.0439} = 2.91 \mbox{ years},
\end{displaymath}
in accordance with the computation of Cochran and Ellner \cite{CochranEllner_1992}.
\end{example} 

\begin{remark}
The matrix $\mathbf{a}$ can be decomposed 
\begin{displaymath}
\mathbf{a} = \mathbf{r} + \mathbf{s},
\end{displaymath}
where $\mathbf{r}$, $\mathbf{s}$ correspond to the reproductive and non reproductive arcs respectively (in the Leslie matrix, the matrix of fertility and survival rates respectively). A convenient formula for computing the  generation time is then
\begin{displaymath}
T = \lambda \frac{\mathbf{v}\mathbf{w}}{\mathbf{v}\mathbf{r}\mathbf{w}}.
\end{displaymath}
In this expression, the terms are matrix products (the denominator is the matrix product of the row vector $\mathbf{v}$, the matrix $\mathbf{r}$, and the column vector $\mathbf{w}$).
\end{remark} 

\begin{remark}
Theorem \ref{main} provides a novel interpretation of the elasticity $e_{ij}$ of the growth rate $\lambda$ to changes in a matrix entry $a_{ij}$. Let us consider a particle traversing the digraph $\mathcal{A}$, i.e., performing a random walk in $\mathcal{P}$. Then by (\ref{eij}), the elasticity $e_{ij}$ is the frequency at which the particle traverses the arc $[j \to i]$.
\end{remark} 

\section{Generation time distribution}
We provide a formula for the distribution of the random variable $\mathcal{T}$ whose expectation is the generation time, $T = \mathbb{E}[\mathcal{T}]$. 

The matrix $\tilde{\mathbf{p}}$ representing the transition probabilities between arcs of $\mathcal{A}$ can be decomposed
\begin{displaymath} 
\tilde{\mathbf{p}} = \left[ \begin{array}{c | c}
    \tilde{\mathbf{p}}_{RR} & \tilde{\mathbf{p}}_{RS} \\
    \hline 
    \tilde{\mathbf{p}}_{SR} & \tilde{\mathbf{p}}_{SS}
  \end{array} \right],
\end{displaymath}
with $R$ the set of reproductive arcs, $S$ the set of non reproductive arcs. In this decomposition, the submatrix $\tilde{\mathbf{p}}_{RR}$ maps $R$ to $R$,  and similar definitions hold for the other submatrices. 

The stationary distribution of $\tilde{\mathbf{p}}$ can also be decomposed
\begin{displaymath} 
\tilde{\boldsymbol{\pi}} = 
\left[
\tilde{\boldsymbol{\pi}}_R \,
\vline \,
\tilde{\boldsymbol{\pi}}_S      
\right].
\end{displaymath}
Using (\ref{1surpi}), the subvectors $\tilde{\boldsymbol{\pi}}_R$ and $\tilde{\boldsymbol{\pi}}_S$ are now scaled so that their entries sum to 1:
\begin{displaymath} 
\boldsymbol{\varpi} = \left[ \boldsymbol{\varpi}_R \, \vline \, \boldsymbol{\varpi}_S \right] = \left[ T_R \tilde{\boldsymbol{\pi}}_R \, \vline \, T_S \tilde{\boldsymbol{\pi}}_S \right].
\end{displaymath}

\begin{proposition}\label{distribT}
The distribution of $\mathcal{T}$ is given by
\begin{displaymath} 
\mathbb{P}[\mathcal{T} = k] = \left\{
\begin{array}{c c}
\boldsymbol{\varpi}_R \, \tilde{\mathbf{p}}_{RR} \, \mathbf{e}, &\quad k = 1,\\
\\
\boldsymbol{\varpi}_R \, \tilde{\mathbf{p}}_{RS} \, (\tilde{\mathbf{p}}_{SS})^{k - 2} \, \tilde{\mathbf{p}}_{SR} \, \mathbf{e},  &\quad k \geqslant 2.\\
\end{array}\right.
\end{displaymath}
Here, $\mathbf{e}$ is a vector of ones of the same dimension as $\boldsymbol{\varpi}_R$.
\end{proposition}
\begin{proof}
When $\mathcal{T} = 1$, we have traveled directly from a reproductive arc to another. The probability of this event is the sum of the probability of being on a reproductive arc, given by the entries of $\boldsymbol{\varpi_R}$, times the probability of going from this arc to another reproductive arc, given by the entries of $\tilde{\mathbf{p}}_{RR}$. When $\mathcal{T} = k \geqslant 2$, starting from an arc of $R$ (probabilities $\boldsymbol{\varpi}_R$), we first go to an arc of $S$ (probabilities $\tilde{\mathbf{p}}_{RS}$) before spending $k - 2$ time intervals on $S$ (probabilities $\tilde{\mathbf{p}}_{SS}$), and then return to $R$ (probabilities $\tilde{\mathbf{p}}_{SR}$).
\end{proof}

\section{Lebreton's formula}\label{lebreton} 

\begin{theorem}\label{c_arcs}
Let $c$ be a common parameter multiplying the entries $a_{ij}$ associated with the reproductive arcs, then
\begin{equation}\label{1surT}
e_\lambda(c) = \frac{1}{T}.
\end{equation}
Let $d$ be a common parameter multiplying the entries $a_{ij}$ associated with the non reproductive arcs, then
\begin{equation}\label{1moins1surT}
e_\lambda(d) = 1 - \frac{1}{T}.
\end{equation}
\end{theorem}
\begin{proof}
For a reproductive arc $[j \to i] \in R$, we have $a_{ij} = c b_{ij}$, hence $\frac{\partial a_{ij}}{\partial c} = b_{ij} = \frac{a_{ij}}{c}$ for $[j \to i] \in R$, and  $\frac{\partial a_{ij}}{\partial c} = 0$ otherwise. Now,
\begin{displaymath}
e_\lambda(c) = \frac{c}{\lambda} \frac{\partial \lambda}{\partial c} = \frac{c}{\lambda} \sum_{i,j} \frac{\partial \lambda}{\partial a_{ij}}\frac{\partial a_{ij}}{\partial c} = \sum_{[j \to i] \in R} \frac{a_{ij}}{\lambda} \frac{\partial \lambda}{\partial a_{ij}} = \sum_{[j \to i ] \in R} e_{ij} = \frac{1}{T},
\end{displaymath}
where the last equality comes from Theorem \ref{main}. For the non reproductive arcs, as the elasticities sum to 1, $e_\lambda(d) = 1 - e_\lambda(c)$, giving (\ref{1moins1surT}).
\end{proof}
Formulas (\ref{1surT}) and (\ref{1moins1surT}) were shown by Lebreton \cite{HoullierLebreton_1986,Lebreton_1996} in the case of the Leslie matrix for a common parameter multiplying the fertilities $f_i$ in the first row, and for a common parameter multiplying the survival rates $s_i$ in the subdiagonal. In the pre-breeding census, fertilities are written $f_i = \sigma s_0 g_i$ with $\sigma$ the primary female sex-ratio, $s_0$ the juvenile survival (from birth to age 1), and $g_i$ the fecundity at age $i$. The demographic parameters $\sigma$ and $s_0$ are common factors of the fertilities.

These formulas have important consequences for life history evolution. For example, short-lived species have small generation time, hence large sensitivity in juvenile survival and primary sex-ratio. By constrast, long-lived species have large generation time, low sensitivity in juvenile survival and large sensitivity in adult survival.

\section{Concluding remarks}\label{conclusion} 
Though we have used a biological formalism to describe the generation time, the setting we have developed is quite general. If a process can be described by a primitive weighted digraph $\mathcal{A}$, and if some arcs of $\mathcal{A}$ can be identified as reproductive in the sense that they lead to the renewal of the entities described by the process, then the generation time can be defined by (\ref{T}). More generally, Theorem \ref{main} provides a way to compute the return time with respect to any property shared by specific transitions of the process.

The generation time can be seen as the mean time by which novelty is brought to a system by its internal dynamics. It remains to explore the consequences of this definition for dynamical systems more general than the linear ones we have considered here. 


\bigskip
\hrule
\bigskip

\noindent 2000 \textit{Mathematics Subject Classification}: Primary 05C38, 92D25; Secondary 05C50, 60J20.\\

\noindent \textit{Keywords}: generation time, life cycle, weighted directed graph, matrix population model, Markov chain, return time, sensitivity, elasticity.

\bigskip
\hrule
\bigskip


\begin{thebibliography}{99}


\bibitem{Caswell_1978a}
Caswell H. 1978. A general formula for the sensitivity of population growth rate to changes in life history parameters. Ecology 59:53--66.

\bibitem{Caswell_1978b}
Caswell H and P Werner. 1978. Transient behavior and life history analysis of teasel (\textit{Dipsacus sylvestris} Huds.). Theoretical Population Biology 14:215--230.

\bibitem{Caswell_2000}
Caswell H. 2000. Matrix Population Models: Construction, Analysis, and Interpretation. 2nd edition, Sinauer, Sunderland, Massachussets.

\bibitem{CochranEllner_1992}
Cochran ME and S Ellner. 1992. Simple methods for calculating age-specific life history parameters from stage-structured models. Ecological Monographs 62:345--364.

\bibitem{Demetrius_1969}
Demetrius L. 1969. The sensitivity of population growth rate to perturbations in the life cycle components. Mathematical Biosciences 4:129--139.

\bibitem{Demetrius_1974}
Demetrius L. 1974. Demographic parameters and natural selection. Proceedings of the National Academy of Sciences USA 71:4645--4647.

\bibitem{Demetrius_1975}
Demetrius L. 1975. Natural selection and age-structured populations. Genetics 79:535--544.

\bibitem{Demetrius_2006}
Demetrius L. 2006. The origin of allometric scaling laws in biology. Journal of Theoretical Biology 243:455–-467.

\bibitem{DemetriusLegendre_2009}
Demetrius L, S Legendre and P Harrem\"{o}es. 2009. Evolutionary entropy: A predictor of body size, metabolic rate and maximal life span. Bulletin of Mathematical Biology 71:800--818.

\bibitem{Demetrius_2013}
Demetrius L. 2013. Boltzmann, Darwin and directionality theory. Physics reports.

\bibitem{Euler_1760}
Euler L. 1760. Recherches g\'{e}n\'{e}rales sur la mortalit\'{e} et la multiplication du genre humain. Histoire de l'Acad\'{e}mie Royale des Sciences
et Belles Lettres de Belgique, 144--164.

\bibitem{HoullierLebreton_1986}
Houllier F and J-D Lebreton. 1986. A renewal equation approach to the dynamics of stage-grouped populations. Mathematical Biosciences 79:185--197.

\bibitem{Lebreton_1996}
Lebreton J-D. 1996. Demographic models for subdivided populations. Theoretical Population Biology 49:291--313.

\bibitem{Leslie_1945}
Leslie PH. 1945. On the use of matrices in population mathematics. Biometrika 33:183--212.

\bibitem{Seneta_2006}
Seneta E. 2006. Non-negative Matrices and Markov Chains. Springer Series in Statistics, Springer, USA.

\bibitem{Tuljapurkar_1982}
Tuljapurkar S. 1982. Why use population entropy? It determines the rate of convergence. Journal of Mathematical Biology 13:325-337.

\end{thebibliography}
\end{document}